\documentclass[11pt]{article}
\usepackage{mathrsfs}
\usepackage[centertags]{amsmath}
\usepackage{amsfonts}
\usepackage{amssymb}
\usepackage{amsthm}
\usepackage{cases}
\usepackage{indentfirst}
\usepackage{epsfig}
\usepackage{graphicx}
\usepackage{color}
\usepackage{CJK}
\usepackage[numbers,sort&compress]{natbib}
\usepackage{subfigure} 
\usepackage[noend]{algpseudocode}
\usepackage{algorithmicx,algorithm}

\date{}

\newtheorem{theorem}{Theorem}[section]

\newtheorem{lemma}{Lemma}[section]

\newtheorem{remark}{Remark}[section]
\newtheorem{definition}{Definition}[section]
\numberwithin{equation}{section}

\allowdisplaybreaks[4]

\begin{document}
\title{\textbf{The perturbation analysis of nonconvex low-rank matrix robust recovery}}
\author{$^{a}$Jianwen Huang \qquad$^{b}$Wendong Wang\qquad $^{b,c}$Feng Zhang \qquad $^{c}$Jianjun Wang\thanks{Corresponding author, E-mail: wjjmath@gmail.com, wjj@swu.edu.cn(J.J. Wang), E-mail: hjw1303987297@126.com (J. Huang)}    \\
{\small $^{a}$School of Mathematics and Statistics, Tianshui Normal University, Tianshui, 741001, China}\\
{\small $^{b}$School of Mathematics and Statistics, Southwest University, Chongqing, 400715, China}\\
{\small $^c$College of Artificial Intelligence, Southwest
University, Chongqing, 400715, China} }

\maketitle
\begin{quote}
{\bf Abstract.}~~In this paper, we bring forward a completely perturbed nonconvex Schatten $p$-minimization to address a model of completely perturbed low-rank matrix recovery. The paper that based on the restricted isometry property generalizes the investigation to a complete perturbation model thinking over not only noise but also perturbation, gives the restricted isometry property condition that guarantees the recovery of low-rank matrix and the corresponding reconstruction error bound. In particular, the analysis of the result reveals that in the case that $p$ decreases $0$ and $a>1$ for the complete perturbation and low-rank matrix, the condition is the optimal sufficient condition $\delta_{2r}<1$ \cite{Recht et al 2010}. The numerical experiments are conducted to show better performance, and provides outperformance of the nonconvex Schatten $p$-minimization method comparing with the convex nuclear norm minimization approach in the completely perturbed scenario.

\end{quote}

\section{Introduction}
\label{sec1}

Low-rank matrix recovery (LMR) is a rapidly developing topic attracting the interest of numerous researchers in the field of optimization and compressed sensing. Mathematically, we can describe it as follows:
\begin{align}\label{eq.1}
y=\mathcal{A}(X)
\end{align}
where $\mathcal{A}:~\mathbb{R}^{m\times n}\to\mathbb{R}^M$ is a known linear transformation (we suppose that $m\leq n$), $y\in\mathbb{R}^M$ is a given observation vector, and $X\in\mathbb{R}^{m\times n}$ is the matrix to be recovered. The objective of LMR is to find the lowest rank matrix based on $(y,~\mathcal{A})$. If the observation $y$ is corrupted by noise $z$, model (\ref{eq.1}) is changed into the following form
\begin{align}\label{eq.2}
\hat{y}=\mathcal{A}(X)+z
\end{align}
where $\hat{y}$ is the noisy measurement, and $z$ is the additive noise independent of the matrix $X$. However, more LMR models can be encountered where not only the linear measurement $y$ is contaminated by the noise vector $z$, but also the linear transformation $\mathcal{A}$ is perturbed by $\mathcal{E}$ for completely perturbed setting, namely, substitute the linear transformation $\mathcal{A}$ with $\mathcal{\hat{A}}=\mathcal{A}+\mathcal{E}$. The completely perturbed appearance arises in remote sensing\cite{Fannjiang et al 2010}, radar\cite{Herman and Strohmer 2009}, source separation\cite{Blumensath and Davies 2007}, etc. When $m=n$ and the matrix $X=\mbox{diag}(x)~(x\in\mathbb{R}^m)$ is diagonal, models (\ref{eq.1}) and (\ref{eq.2}) degenerates to the compressed sensing models
\begin{align}
\label{eq.3}y&=Ax,\\
\label{eq.4}\hat{y}&=Ax+z
\end{align}
where $A\in\mathbb{R}^{M\times m}$ is a measurement matrix and $x\in\mathbb{R}^m$ is an unknown sparse signal. We call the problem (\ref{eq.3}) as the sparse signal recovery. For the completely perturbed model, the convex nuclear norm minimization is frequently considered \cite{Huang et al 2019} as follows:
\begin{align}\label{eq.5}
\min_{\tilde{Z}\in\mathbb{R}^{m\times n}}\|\tilde{Z}\|_*~\mbox{s.t.}~\|\hat{\mathcal{A}}(\tilde{Z})-\hat{y}\|_2\leq\epsilon'_{\mathcal{A},r,y},
\end{align}
where $\|\tilde{Z}\|_*$ is the nuclear norm of the matrix $\tilde{Z}$, that is, the sum of its singular values, and $\epsilon'_{\mathcal{A},r,y}$ is the total noise level. Problem (\ref{eq.5}) can be reduced to the $l_1$-minimization \cite{Herman and Strohmer 2010}
\begin{align}\label{eq.6}
\min_{\tilde{z}\in\mathbb{R}^{n_1}}\|\tilde{z}\|_1~\mbox{s.t.}~\|\hat{A}\tilde{z}-\hat{y}\|_2\leq\epsilon'_{A,r,y},
\end{align}
where $\|\tilde{z}\|_1$ is the $l_1$-norm of the vector $\tilde{z}$, that is, the sum of absolute value of its coefficients.

Chartrand \cite{Chartrand 2007} showed that fewer measurements are required for exact reconstruction if $l_1$-norm is substituted with $l_p$-norm. There exist many work regarding reconstructing $x$ via the $l_p$-minimization \cite{Chartrand and Staneva 2008}, \cite{Foucart and Lai 2009}, \cite{Lai and Liu 2011}, \cite{Lai et al 2013}, \cite{Wang Y et al 2014}, \cite{Song and Xia 2014}, \cite{Wang J et al 2015}, \cite{Wen J et al 2015}, \cite{Wen F et al 2017}, \cite{Gao et al 2017}, \cite{Zhang and Li 2017}, \cite{Wen F et al 2018}. In \cite{Chartrand 2007}, numerical simulations demonstrated that fewer measurements are needed for exact reconstruction than when $p=1$.

In this paper, we are interested in the completely perturbed model for the nonconvex Schatten $p$-minimization ($0<p<1$)
\begin{align}\label{eq.7}
\min_{\tilde{Z}\in\mathbb{R}^{m\times n}}\|\tilde{Z}\|^p_p~\mbox{s.t.}~\|\hat{\mathcal{A}}(\tilde{Z})-\hat{y}\|_2\leq\epsilon'_{\mathcal{A},r,y},
\end{align}
where $\|\tilde{Z}\|^p_p$ is the Schatten $p$ quasi-norm of the matrix $\tilde{Z}$, that is, $\|\tilde{Z}\|^p_p=(\sum_i\sigma^p_i(\tilde{Z}))^{1/p}$ with $\sigma_i(\tilde{Z})$ being $i$th singular value of $\tilde{Z}$. Problem (\ref{eq.7}) can be returned to the $l_p$-minimization \cite{Ince and Nacaroglu 2014}
\begin{align}\label{eq.8}
\min_{\tilde{z}\in\mathbb{R}^{M\times m}}\|\tilde{z}\|^p_p~\mbox{s.t.}~\|\hat{A}\tilde{z}-\hat{y}\|_2\leq\epsilon'_{A,r,y},
\end{align}
where $\|\tilde{z}\|^p_p=(\sum_i\tilde{z}^p_i)^{1/p}$ is the $l_p$-quasi-norm of the vector $\tilde{z}$. To the best of our knowledge, recently researches are considered only in unperturbed situation ($\mathcal{E}=0$), that is, the linear transformation $\mathcal{A}$ is not perturbed by $\mathcal{E}$ (for related work, see \cite{Mohan and Fazel 2010}, \cite{Dvijotham and Fazel 2010}, \cite{Zhang et al 2013}, \cite{Kong and Xiu 2013}, \cite{Wang and Li 2013}, \cite{Chen and Li 2015}, \cite{Gao Y et al 2017b}, \cite{Wang W et al 2019}). From the perspective of application, it is more practical to investigate the recovery of low-rank matrices in the scenario of complete perturbation.

In this paper, based on restricted isometry property (RIP), the performance of low-rank matrices reconstruction is showed by the nonconvex Schatten $p$-minimization in completely perturbed setting. The main contributions of this paper are as follows. First, we present a sufficient condition for reconstruction of low-rank matrices via the nonconvex Schatten $p$-minimization. Second, the estimation accurateness between the optimal solution and the original matrix is described by a total noise and a best $r$-rank approximation error. The result reveals that stable and robust performance concerning reconstruction of low-rank matrices in existence of total noise. Third, numerical experiments are conducted to sustain the gained results, and demonstrate that the performance of nonconvex Schatten $p$-minimization can be better than that of convex nuclear norm minimization in completely perturbed model.

The rest of this paper is constructed as follows.

\section{Notation and main results}

Before presenting the main results, we first introduce the notion of RIC of a linear transformation $\mathcal{A}$, which is as follows.
\begin{definition}\label{def.1}
The restricted isometry constant (RIC) $\delta_r$ of a linear transformation $\mathcal{A}$ is the smallest constant such that
\begin{align}\label{eq.9}
(1-\delta)\|X\|^2_F\leq\|\mathcal{A}(X)\|_2^2\leq (1+\delta)\|X\|^2_F
\end{align}
holds for all $r$-rank $X\in\mathbb{R}^{m\times n}$ (i.e., $rank(X)\leq r$), where $\|X\|_F:=\sqrt{\left<X,X\right>}=\sqrt{trace(X^{\top}X)}$ is the Frobenius norm of the matrix $X$.
\end{definition}
Then we provide some notations similar to \cite{Huang et al 2019}, which quantifying the perturbations $\mathcal{E}$ and $z$ with the bounds:
\begin{align}\label{eq.10}
\frac{\|\mathcal{E}\|_{op}}{\|\mathcal{A}\|_{op}}\leq\epsilon_{\mathcal{A}},~
\frac{\|\mathcal{E}\|^{(r)}_{op}}{\|\mathcal{A}\|^{(r)}_{op}}\leq\epsilon^{(r)}_{\mathcal{A}},~
\frac{\|z\|_2}{\|y\|_2}\leq\epsilon_{y},
\end{align}
where $\|\mathcal{A}\|_{op}=\sup\{\|\mathcal{A}(X)\|_2/\|X\|_F:~X\in\mathbb{R}^{m\times n}\setminus\{0\}\}$ is the operator norm of linear transformation $\mathcal{A}$, and $\|\mathcal{A}\|^{(r)}_{op}=\sup\{\|\mathcal{A}(X)\|_2/\|X\|_F:~X\in\mathbb{R}^{m\times n}\setminus\{0\}~\mbox{and}~\mbox{rank}(X)\leq r\}$, and representing
\begin{align}\label{eq.11}
t_r=\frac{\|X_{[r]^c}\|_F}{\|X_{[r]}\|_F},~s_r=\frac{\|X_{[r]^c}\|_*}{\sqrt{r}\|X_{[r]}\|_F},~
\kappa^{(r)}_{\mathcal{A}}=\frac{\sqrt{1+\delta_r}}{\sqrt{1-\delta_r}},~
\alpha_{\mathcal{A}}=\frac{\|\mathcal{A}\|_{op}}{\sqrt{1-\delta_r}}.
\end{align}
Here $X_{[r]}$ is the best $r$-rank approximation of the matrix $X$, its singular values are composed of $r$-largest singular values of the matrix $X$, and $X_{[r]^c}=X-X_{[r]}$. With notations and symbols above, we present our results for reconstruction of low-rank matrices via the completely perturbed nonconvex Schatten $p$-minimization.
\begin{theorem}\label{the.1}
For given relative perturbations $\epsilon_{\mathcal{A}}$, $\epsilon^{(r)}_{\mathcal{A}}$, $\epsilon^{(2r)}_{\mathcal{A}}$, and $\epsilon_{y}$ in (\ref{eq.10}), suppose the RIC for the linear transformation $\mathcal{A}$ fulfills
\begin{align}\label{eq.12}
\delta_{2ar}<\frac{2+\sqrt{2}a^{1/2-1/p}}{(1+\sqrt{2}a^{1/2-1/p})(1+\epsilon^{(2ar)}_{\mathcal{A}})^2}-1
\end{align}
for $a>1$ and
that the general matrix $X$ meets
\begin{align}\label{eq.13}
t_r+s_r<\frac{1}{\kappa^{(r)}_{\mathcal{A}}}.
\end{align}
Then a minimizer $X^*$ of problem (\ref{eq.7}) approximates the true matrix $X$ with errors
\begin{align}
\label{eq.14}\|X-X^*\|^p_F&\leq C_1(\epsilon'_{\mathcal{A},r,y})^p+C_2\frac{\|X_{[r]^c}\|^p_p}{r^{1-p/2}},\\
\label{eq.15}\|X-X^*\|^p_p&\leq C'_1r^{1-p/2}(\epsilon'_{\mathcal{A},r,y})^p+C'_2\|X_{[r]^c}\|^p_p,
\end{align}
where the total noise is
\begin{align}\label{eq.16}
\epsilon'_{\mathcal{A},r,y}=\left[\frac{\epsilon^{(r)}_{\mathcal{A}}\kappa^{(r)}_{\mathcal{A}}
+\epsilon_{\mathcal{A}}\alpha_{\mathcal{A}}t_r}{1-\kappa^{(r)}_{\mathcal{A}}(t_r+s_r)}+\epsilon_{y}\right]\|y\|_2,
\end{align}
and
\begin{align}
\label{eq.17}
C_1&=\frac{2^{p}(1+a^{p/2-1})
(1+\hat{\delta}_{(a+1)r})^{p/2}}{(1-\hat{\delta}_{(a+1)r})^{p}-a^{p/2-1}(\hat{\delta}_{(a+1)r}^2+\hat{\delta}_{2ar}^2)^{p/2}},\\
\label{eq.18}
C_2&=2a^{p/2-1}[1+\frac{(1+a^{p/2-1})
(\hat{\delta}_{(a+1)r}^2+\hat{\delta}_{2ar}^2)^{p/2}}{(1-\hat{\delta}_{(a+1)r})^{p}-a^{p/2-1}(\hat{\delta}_{(a+1)r}^2+\hat{\delta}_{2ar}^2)^{p/2}}],\\
\label{eq.19}
C'_1&=\frac{2^{p+1}(1+a)^{1-p/2}
(1+\hat{\delta}_{(a+1)r})^{p/2}}{(1-\hat{\delta}_{(a+1)r})^{p}-a^{p/2-1}(\hat{\delta}_{(a+1)r}^2+\hat{\delta}_{2ar}^2)^{p/2}},\\
\label{eq.18}
C'_2&=2+\frac{4(1+a)^{1-p/2}a^{p/2-1}
(\hat{\delta}_{(a+1)r}^2+\hat{\delta}_{2ar}^2)^{p/2}}{(1-\hat{\delta}_{(a+1)r})^{p}-a^{p/2-1}(\hat{\delta}_{(a+1)r}^2+\hat{\delta}_{2ar}^2)^{p/2}},
\end{align}
where $\hat{\delta}_{(a+1)r}=(1+\delta_{(a+1)r})(1+\epsilon_{A}^{((a+1)r)})^2-1,~~\hat{\delta}_{2ar}=(1+\delta_{2ar})(1+\epsilon_{A}^{(2ar)})^2-1$.
\end{theorem}

\begin{remark}
Theorem \ref{the.1} gives a sufficient conditions for the reconstruction of low-rank matrices via nonconvex Schatten $p$-minimization in completely perturbed scenario. Condition (\ref{eq.12}) of the Theorem extends the assumption of $l_p$ situation in \cite{Ince and Nacaroglu 2014} to the nonconvex Schatten $p$-minimization. Observe that as the value of $p$ becomes large, the bound of RIC $\delta_{2ar}$ reduces, which reveals that smaller value of $p$ can induce weaker reconstruction guarantee. Particularly, when $p\to 0~(a>1)$ ((\ref{eq.7}) degenerates to the rank minimization: $\min_{\tilde{Z}\in\mathbb{R}^{m\times n}}\mbox{rank}(\tilde{Z})~\mbox{s.t.}~\|\hat{\mathcal{A}}(\tilde{Z})-\hat{y}\|_2\leq\epsilon'_{\mathcal{A},r,y}$), it leads to the RIP condition $\delta_{2r}<2/(1+\epsilon^{(2ar)}_{\mathcal{A}})^2-1$ for reconstruction of low-rank matrices via the rank minimization, to the best of our knowledge, the current optimal recovery condition about RIP is $\delta_{2r}<1$ to ensure exact reconstruction for $r$-rank matrices via rank minimization \cite{Recht et al 2010}, therefore the Theorem extends that condition to the scenario of presence of noise and $r$-rank matrices. Furthermore, when $m=n$ and the matrix $X=\mbox{diag}(x)~(x\in\mathbb{R}^m)$ is diagonal, the Theorem reduces to the case of compressed sensing given by \cite{Ince and Nacaroglu 2014}.
\end{remark}

\begin{remark}
Under the requirement (\ref{eq.12}), one can easily check that the condition (\ref{eq.13}) is satisfied. Besides, when $\mbox{rank}(X)\leq r$, the condition (\ref{eq.13}) holds. Additionally, the inequalities (\ref{eq.14}) and (\ref{eq.15}) in Theorem \ref{the.1} which exploit two kinds of metrics provide upper bound estimations on the reconstruction of nonconvex Schatten $p$-minimization. The estimations evidence that reconstruction accurateness can be controlled by the best $r$-rank approximation error and the total noise. In particular, when there aren't noise (i.e., $\mathcal{E}=0$ and $z=0$), they clear that the $r$-rank matrix can be accurately reconstructed via the nonconvex Schatten $p$-minimization. In (\ref{eq.14}), both the error bound noise constant $C_1$ and the error bound compressibility constant $C_2$ may rely on the value of $p$. Numerical simulations reveal that when we fix the other independent parameters, a smaller value of $p$ will produce a smaller $C_1$ and a smaller $C_2/r^{1-p/2}$. For more details, see Fig. \ref{fig.1}.

\begin{figure}[ht]
\begin{center}
\subfigure[]{\includegraphics[width=0.40\textwidth]{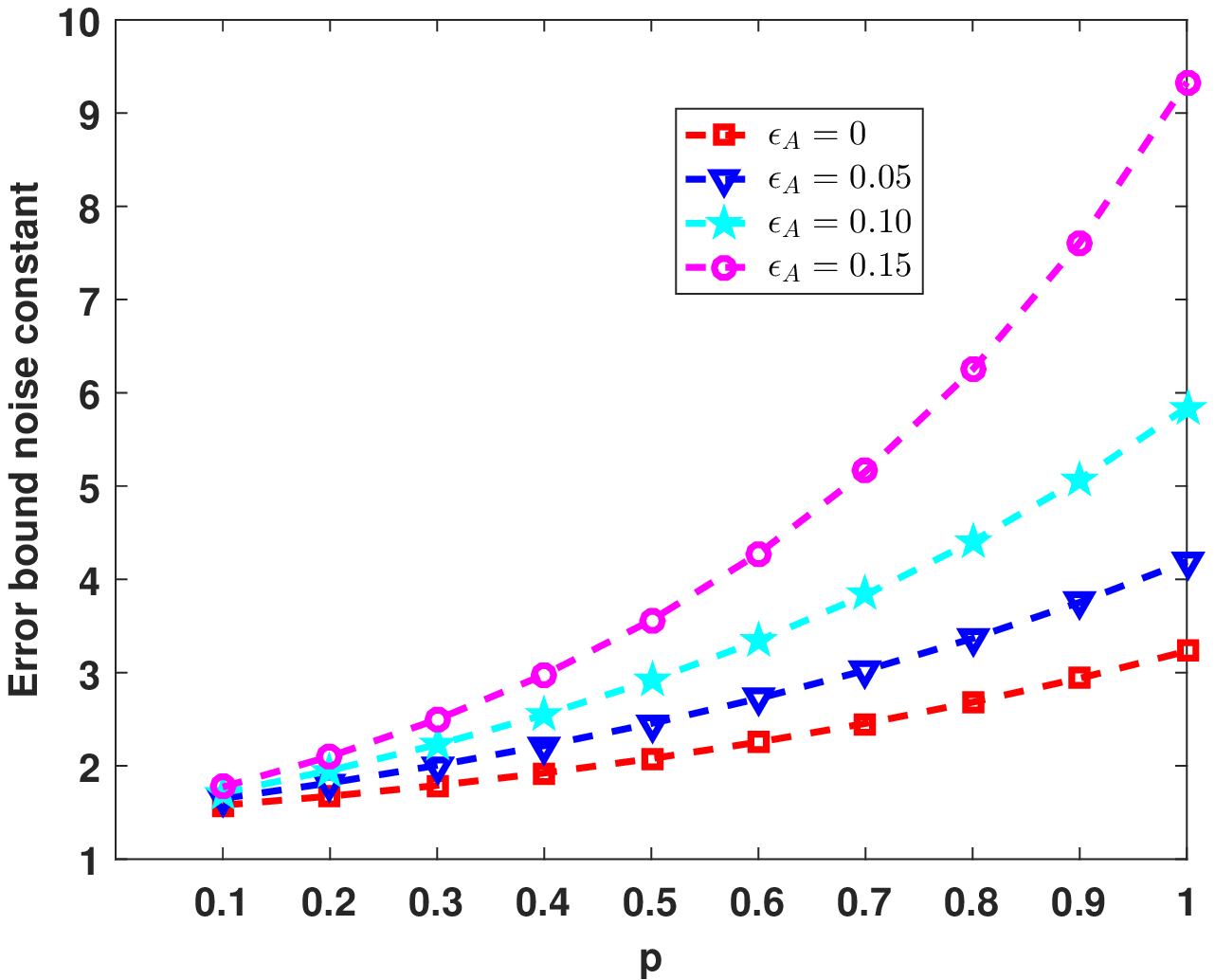}}
\subfigure[]{\includegraphics[width=0.40\textwidth]{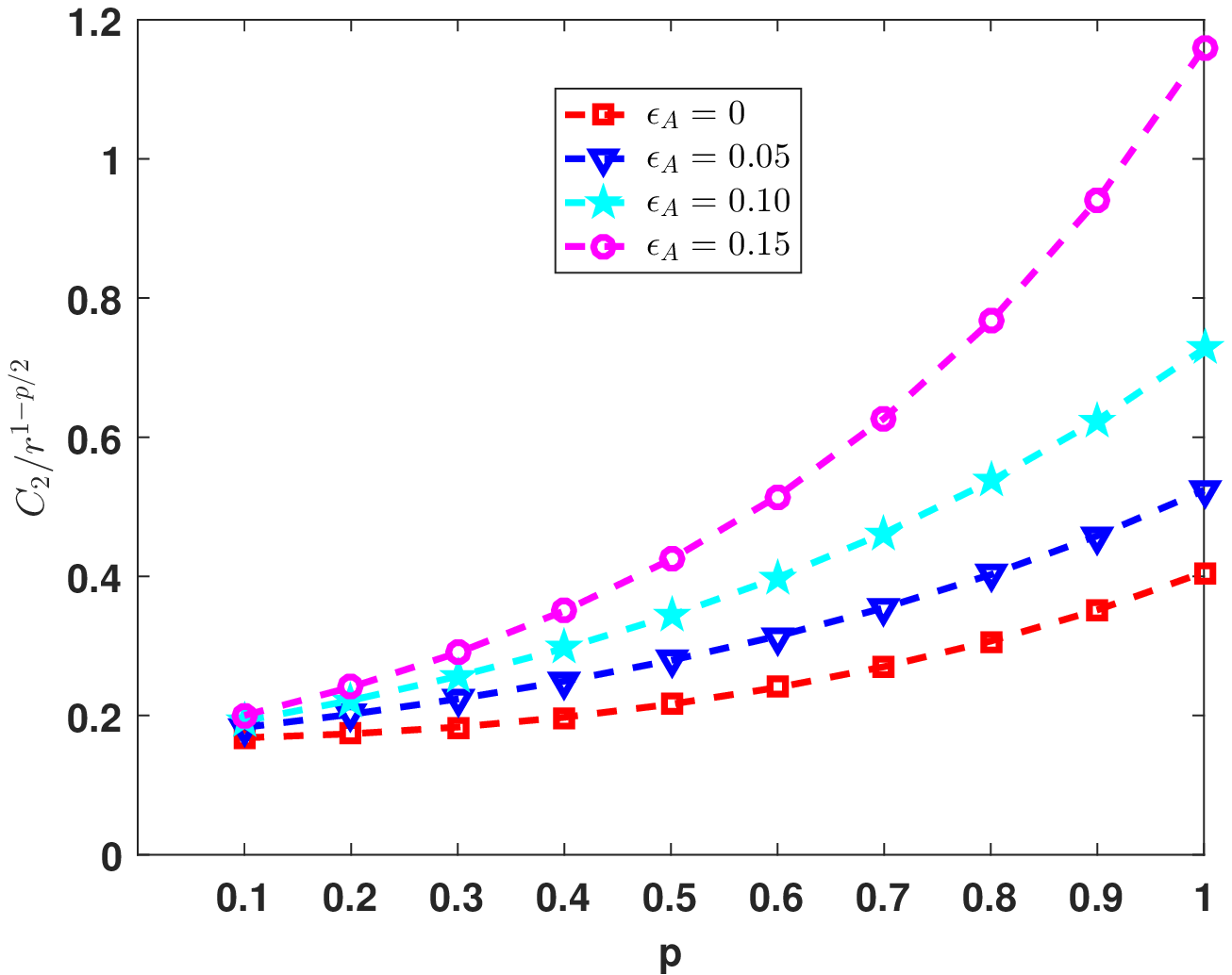}}
\caption{The error bound noise constant $C_1$, and the error bound compressibility constant $C_2/r^{1-p/2}$ versus $p$ for $a=5, \delta_{(a+1)r}=\delta_{2ar}=0.05$ in (a) and (b), respectively.}\label{fig.1}
\end{center}
\vspace*{-14pt}
\end{figure}
\end{remark}

\begin{remark}
When the matrix $X$ is a strictly $r$-rank matrix (i.e., $X=X_{[r]}$), a minimizer $X^*$ of problem (\ref{eq.7}) approximates the true matrix $X$ with errors
\begin{align}
\notag\|X-X^*\|_F&\leq C^{1/p}_1\epsilon'_{\mathcal{A},r,y},\\
\notag\|X-X^*\|_p&\leq C'^{1/p}_1r^{1/p-1/2}\epsilon'_{\mathcal{A},r,y},
\end{align}
where
\begin{align}
\notag\epsilon'_{\mathcal{A},r,y}=[\epsilon^{(r)}_{\mathcal{A}}\kappa^{(r)}_{\mathcal{A}}+\epsilon_{y}]\|y\|_2.
\end{align}
In the case of $\mathcal{E}=0$, that is, there doesn't exist perturbation in the linear transformation $\mathcal{A}$, then $\hat{\delta}_{(a+1)r}=\delta_{(a+1)r},$ $\hat{\delta}_{2ar}=\delta_{2ar}$. In the case that $m=n$, the matrix $X=\mbox{diag}(x)~(x\in\mathbb{R}^m)$ is diagonal (i.e., the results of Theorem reduce to the case of compressed sensing), $p=1$ and $a=1$, our result contains that of Theorem $2$ in \cite{Herman and Strohmer 2010}.
\end{remark}

\section{Proofs of the main results}

In this part, we will provide the proofs of main results. In order to prove our main results, we need the following auxiliary lemmas. Firstly, we give Lemma \ref{lem.1} which  incorporates an important inequality associating with $\delta_r$ and $\hat{\delta}_r$.

\begin{lemma}\label{lem.1}(RIP for $\mathcal{\hat{A}}$ \cite{Huang et al 2019}) Given the RIC $\delta_r$ related with linear transformation $\mathcal{A}$ and the relative perturbation $\epsilon^{(r)}_{\mathcal{A}}$ corresponded with linear transformation $\mathcal{E}$, fix the constant $\hat{\delta}_{r,\max}=(1+\delta_r)(1+\epsilon^{(r)}_{\mathcal{A}})^2-1$. Then the RIC $\hat{\delta}_r\leq\hat{\delta}_{r,\max}$ for $\mathcal{\hat{A}}=\mathcal{A}+\mathcal{E}$ is the smallest nonnegative constant such that
\begin{align}\label{eq.20}
(1-\hat{\delta}_r)\|X\|^2_F\leq\|\mathcal{\hat{A}}(X)\|_2^2\leq (1+\hat{\delta}_r)\|X\|^2_F
\end{align}
holds for all matrices $X\in\mathbb{R}^{m\times n}$ that are $r$-rank.
\end{lemma}

We will employ the fact that $\mathcal{\hat{A}}$ maps low-rank orthogonal matrices to nearly sparse orthogonal vectors, which is given by \cite{Candes and Plan 2011}.
\begin{lemma}\label{lem.2}(\cite{Candes and Plan 2011})
For all $X,~Y$ satisfying $\left<X,Y\right>=0$, and $\mbox{rank}(X)\leq r_1$, $\mbox{rank}(Y)\leq r_2$,
\begin{align}\label{eq.22}
\left|\left<\mathcal{\hat{A}}(X),\mathcal{\hat{A}}(Y)\right>\right|\leq\hat{\delta}_{r_1+r_2}\|X\|_F\|Y\|_F.
\end{align}
\end{lemma}
Moreover, the following lemma will be utilized in the proof of main result, which combines with Lemma $2.3$ \cite{Recht et al 2010} and Lemma $2.2$ \cite{Kong and Xiu 2013}.
\begin{lemma}\label{lem.3}
Assume that $X,~Y\in\mathbb{R}^{m\times n}$ obey $X^{\top}Y=0$ and $XY^{\top}=0$. Let $0<p\leq 1$. Then
\begin{align}\label{eq.23}
\|X+Y\|^p_p=\|X\|^p_p+\|Y\|^p_p,~\|X+Y\|_p\geq\|X\|_p+\|Y\|_p,
\end{align}
where $\|X\|^p_p$ and $\|X\|_p$ stand for the nuclear norm of matrix $X$ in the case of $p=1$.
\end{lemma}

For any matrix $X\in\mathbb{R}^{m\times n}$, we represent the singular values decomposition (SVD) of $X$ as
\begin{align}\notag
X=U\mbox{diag}(\sigma(X))V^{\top},
\end{align}
where $\sigma(X):=(\sigma_1(X),\cdots,\sigma_m(X))$ is the vector of the singular values of $X$, $U$ and $V$ are respectively the left and right singular value matrices of $X$.

\begin{proof}[Proof of the theorem \ref{the.1}]
Let $X$ denote the original matrix to be recovered and $X^*$ denote the optimal solution of (\ref{eq.7}). Let $Z=X-X^*$, and based on the SVD of $X$, its  SVD is given by
\begin{align}\notag
U^{\top}ZV=U_1\mbox{diag}(\sigma(U^{\top}ZV))V^{\top}_1,
\end{align}
where $U_1,~V_1\in\mathbb{R}^{m\times m}$ are orthogonal matrices, and $\sigma(U^{\top}ZV)$ stands for the vector comprised of the singular values of $U^{\top}ZV$. Let $T_0$ is the set composed of the locations of the $r$ largest magnitudes of elements of $\sigma(X)$. We adopt technology similar to the reference \cite{Ince and Nacaroglu 2014} to partition $\sigma(U^{\top}ZV)$ into a sum of vectors $\sigma_{T_i}(U^{\top}ZV)~(i=0,1,\cdots,J)$, where $T_1$ is the set composed of the locations of the $ar$ largest magnitudes of entries of $\sigma_{T^c_0}(U^{\top}ZV)$, $T_2$ is the set composed of the locations of the second $ar$ largest magnitudes of entries of $\sigma_{T^c_0}(U^{\top}ZV)$, and so forth (except possibly $T_J$). Then $Z=\sum_{i=0}^JZ_{T_i}$ where $Z_{T_i}=UU_1\mbox{diag}(\sigma_{T_i}(U^{\top}ZV))(VV_1)^{\top}$, $i=0,1,\cdots,J$. One can easily verify that $Z^{\top}_{T_i}Z_{T_j}=0$ and $Z_{T_i}Z^{\top}_{T_j}=0$ for all $i\neq j$, and $\mbox{rank}(Z_{T_0})\leq r$, $\mbox{rank}(Z_{T_j})\leq ar$, $i=0,1,\cdots,J$. For  simplicity, denote $T_{01}=T_0\bigcup T_1$. Then, we have (see (22) in \cite{Zhang et al 2013}, Lemma $2.6$ \cite{Chen and Li 2015})
\begin{align}\label{eq.24}
\|Z_{T^c_0}\|^p_p\leq\|Z_{T_0}\|^p_p+2\|X_{[r]^c}\|^p_p.
\end{align}
By the decomposition of $Z$, for each $l\in T_i,~k\in T_{i-1},~i\geq 2$, $\sigma_{T_{i}}(U^{\top}ZV)[l]\leq\sigma_{T_{i-1}}(U^{\top}ZV)[k]$, it implies that
\begin{align}\label{eq.25}
(\sigma_{T_{i}}(U^{\top}ZV)[l])^p\leq \frac{\sum^{ar}_{k=1}(\sigma_{T_{i-1}}(U^{\top}ZV)[l])^p}{ar}=
\frac{\|\sigma_{T_{i-1}}(U^{\top}ZV)\|^p_p}{ar}=\frac{\|Z_{T_{i-1}}\|^p_p}{ar},
\end{align}
which deduces
\begin{align}\label{eq.26}
\|Z_{T_{i}}\|^2_F\leq(ar)^{1-\frac{2}{p}}\|Z_{T_{i-1}}\|^2_p.
\end{align}
Thereby,
\begin{align}\label{eq.27}
\|Z_{T_{i}}\|^p_F\leq(ar)^{\frac{p}{2}-1}\|Z_{T_{i-1}}\|^p_p.
\end{align}
Notice that $Z^{\top}_{T_i}Z_{T_j}=0$ and $Z_{T_i}Z^{\top}_{T_j}=0$ for all $i\neq j$, due to Lemma \ref{lem.3} and (\ref{eq.27}), then we can get
\begin{align}\label{eq.28}
\sum_{i\geq2}\|Z_{T_{i}}\|^p_F\leq(ar)^{\frac{p}{2}-1}\sum_{i\geq2}\|Z_{T_{i-1}}\|^p_p
=(ar)^{\frac{p}{2}-1}\|Z_{T^c_0}\|^p_p.
\end{align}
By the inequality $\|Z_{T_0}\|^p_F\leq\|Z_{T_{01}}\|^p_F$ and H$\ddot{o}$lder's inequality, we get
\begin{align}\label{eq.29}
\|Z_{T_0}\|^p_p\leq r^{1-\frac{p}{2}}\|Z_{T_{01}}\|^p_F.
\end{align}
From (\ref{eq.24}), (\ref{eq.28}), (\ref{eq.29}) and the inequality that for every fixed $n\in\mathbb{N}$, and any $0<\alpha\leq1$, $(\sum_{i=1}^nx)^{\alpha}\leq\sum_{i=1}^nx^{\alpha}$ for every $x_i\geq0,~i=1,\cdots,n$, it follows
\begin{align}\label{eq.30}
\|Z_{T^c_{01}}\|^p_F=(\sum_{i\geq2}\|Z_{T_{i}}\|^2_F)^{\frac{p}{2}}\leq\sum_{i\geq2}\|Z_{T_{i}}\|^p_F
\leq(ar)^{\frac{p}{2}-1}(r^{1-\frac{p}{2}}\|Z_{T_{01}}\|^p_F+2\|X_{[r]^c}\|^p_p).
\end{align}
Since
\begin{align}
\notag\|\mathcal{\hat{A}}(Z_{T_{01}})\|^2_2&=<\hat{\mathcal{A}}(Z_{T_{01}}),\mathcal{\hat{A}}(Z_{T_{01}})>\\
\notag&=<\hat{\mathcal{A}}(Z_{T_{01}}),\mathcal{\hat{A}}(Z)>
-<\hat{\mathcal{A}}(Z_{T_{01}}),\sum_{i\geq2}\mathcal{\hat{A}}(Z_{T_i})>\\
\label{eq.31}&\leq\|\hat{\mathcal{A}}(Z_{T_{01}})\|_2\|\hat{\mathcal{A}}(Z)\|_2
+\sum_{i\geq2}|<\hat{\mathcal{A}}(Z_{T_{01}}),\mathcal{\hat{A}}(Z_{T_i})>|,
\end{align}
we get
\begin{align}\label{eq.32}
\|\mathcal{\hat{A}}(Z_{T_{01}})\|^{2p}_2\overset{\text{(a)}}{\leq}\|\hat{\mathcal{A}}(Z_{T_{01}})\|^p_2\|\hat{\mathcal{A}}(Z)\|^p_2
+\sum_{i\geq2}|<\hat{\mathcal{A}}(Z_{T_{01}}),\mathcal{\hat{A}}(Z_{T_i})>|^p,
\end{align}
where (a) follows from the fact that $(a+b)^p\leq a^p+b^p$ for nonnegative $a$ and $b$.

Additionally, by the minimality of $X^*$, we get
\begin{align}\label{eq.33}
\|\mathcal{\hat{A}}(Z)\|^2_2\leq\|\hat{y}-\mathcal{\hat{A}}(X)\|^2_2
+\|\hat{y}-\mathcal{\hat{A}}(X^*)\|^2_2\leq2\epsilon'_{\mathcal{A},r,y}.
\end{align}
Since $Z_{T_{01}}$ is $(a+1)r$-rank and $Z_{T_i}$ is $ar$-rank, $i\geq2$, by applying the RIP of $\mathcal{\hat{A}}$ and combination with (\ref{eq.32}) and (\ref{eq.33}), we get
\begin{align}\label{eq.34}
\|\mathcal{\hat{A}}(Z_{T_{01}})\|^{2p}_2\leq(2\epsilon'_{\mathcal{A},r,y})^p
(1+\hat{\delta}_{(a+1)r})^{\frac{p}{2}}\|Z_{T_{01}}\|^p_F
+\sum_{i\geq2}|<\hat{\mathcal{A}}(Z_{T_{01}}),\mathcal{\hat{A}}(Z_{T_i})>|^p.
\end{align}
Because $<Z_{T_i},Z_{T_j}>$ for all $i\neq j$, and $Z_{T_0}$ is $r$-rank, by Lemma \ref{lem.2} and (\ref{eq.30}), we get
\begin{align}
\notag\|\mathcal{\hat{A}}(Z_{T_{01}})\|^{2p}_2
&\leq(2\epsilon'_{\mathcal{A},r,y})^p
(1+\hat{\delta}_{(a+1)r})^{\frac{p}{2}}\|Z_{T_{01}}\|^p_F
+(\hat{\delta}_{(a+1)r}\|Z_{T_0}\|_F+\hat{\delta}_{2ar}\|Z_{T_1}\|_F)^p\
\sum_{i\geq2}\|Z_{T_i}\|^p_F\\
\notag&\leq(2\epsilon'_{\mathcal{A},r,y})^p
(1+\hat{\delta}_{(a+1)r})^{\frac{p}{2}}\|Z_{T_{01}}\|^p_F\\
\label{eq.35}&\quad+(\hat{\delta}_{(a+1)r}\|Z_{T_0}\|_F+\hat{\delta}_{2ar}\|Z_{T_1}\|_F)^p
(ar)^{\frac{p}{2}-1}(r^{1-\frac{p}{2}}\|Z_{T_{01}}\|^p_F+2\|X_{[r]^c}\|^p_p)
\end{align}
From (\ref{eq.12}), one can easily check that
\begin{align}\label{eq.36}
a^{\frac{p}{2}-1}(\hat{\delta}_{(a+1)r}^2+\hat{\delta}_{2ar}^2)^{\frac{p}{2}}
<(1-\hat{\delta}_{(a+1)r})^{\frac{p}{2}}.
\end{align}
By (\ref{eq.35}), (\ref{eq.36}) and the inequality
$\|\mathcal{\hat{A}}(Z_{T_{01}})\|^{p}_2\geq(1-\hat{\delta}_{(a+1)r})^{\frac{p}{2}}\|Z_{T_{01}}\|^p_F$, one can get
\begin{align}
\notag\|Z_{T_{01}}\|^p_F\leq&\frac{2^p(1+a^{\frac{p}{2}-1})(1+\hat{\delta}_{(a+1)r})^{\frac{p}{2}}}
{(1-\hat{\delta}_{(a+1)r})^{\frac{p}{2}}-a^{\frac{p}{2}-1}(\hat{\delta}_{(a+1)r}^2+\hat{\delta}_{2ar}^2)^{\frac{p}{2}}}
(\epsilon'_{\mathcal{A},r,y})^p\\
\notag&+\frac{2a^{\frac{p}{2}-1}(\hat{\delta}_{(a+1)k}^2+\hat{\delta}_{2ar}^2)^{\frac{p}{2}}}
{(1-\hat{\delta}_{(a+1)r})^{\frac{p}{2}}-a^{\frac{p}{2}-1}(\hat{\delta}_{(a+1)r}^2+\hat{\delta}_{2ar}^2)^{\frac{p}{2}}}
\frac{\|X_{[r]^c}\|^p_p}{r^{1-\frac{p}{2}}}\\
\label{eq.37}=&:\beta(\epsilon'_{\mathcal{A},r,y})^p+\gamma\frac{\|X_{[r]^c}\|^p_p}{r^{1-\frac{p}{2}}},
\end{align}
consequently,
\begin{align}
\label{eq.38}\|Z_{T_0}\|^p_p&\leq r^{1-\frac{p}{2}}\|Z_{T_0}\|^p_F\\
\notag&\leq\beta r^{1-\frac{p}{2}}(\epsilon'_{\mathcal{A},r,y})^p+\gamma\|X_{[r]^c}\|^p_p.
\end{align}
Thus, from (\ref{eq.30}) and (\ref{eq.37}), we get
\begin{align}
\notag\|Z\|^p_F&\leq\|Z_{T_{01}}\|^p_F+\|Z_{T^c_{01}}\|^p_F\\
\label{eq.39}&\leq C_1(\epsilon'_{\mathcal{A},r,y})^p+C_2\frac{\|X_{[r]^c}\|^p_p}{r^{1-\frac{p}{2}}};
\end{align}
in addition, a combination of (\ref{eq.24}) and (\ref{eq.38}), one can get
\begin{align}
\notag\|Z\|^p_p&\leq\|Z_{T_0}\|^p_p+\|Z_{T^c_0}\|^p_p\\
\label{eq.40}&\leq C'_1r^{1-\frac{p}{2}}(\epsilon'_{\mathcal{A},r,y})^p+C'_2\|X_{[r]^c}\|^p_p,
\end{align}
where the constants $C_1$, $C_2$, $C'_1$ and $C'_2$ are defined in Theorem \ref{the.1}. The proof is complete.
\end{proof}

\section{Numerical experiments}

In this section, we carry out some numerical experiments to sustain verification of our theoretical results, we implement all experiments in MATLAB 2016a running on a PC with an Inter core i7 processor (3.6 GHz) with 8 GB RAM. In order to address the completely perturbed nonconvex Schatten $p$-minimization model, we employ the alternating direction method of multipliers (ADMM) method, which is often applied in compressed sensing and sparse approximation \cite{Lu C et al 2018}, \cite{Wen F et al TSP 2017}, \cite{Wang W et al 2017}, \cite{Wang and Wang 2018}. The constrained optimization problem (\ref{eq.7}) can be transformed into an equivalent unconstrained form
\begin{align}\label{eq.41}
\min_{\tilde{Z}\in\mathbb{R}^{m\times n}}\lambda\|\tilde{Z}\|^p_p+\frac{1}{2}\|\hat{A}\mbox{vec}(\tilde{Z})-\hat{y}\|^2_2,
\end{align}
where $\hat{A}\in\mathbb{R}^{M\times mn}$, $\mbox{vec}(\tilde{Z})$ represents the vectorization of $\tilde{Z}$. Hence, $\hat{A}\mbox{vec}(\tilde{Z})$ presents the linear map $\hat{\mathcal{A}}(\tilde{Z})$. Then, introducing an auxiliary variable $W\in\mathbb{R}^{m\times n}$, the problem (\ref{eq.41}) can be equivalently turned into

\begin{align}\label{eq.42}
\min_{W,~\tilde{Z}\in\mathbb{R}^{m\times n}}\lambda\|W\|^p_p+\frac{1}{2}\|\hat{A}\mbox{vec}(\tilde{Z})-\hat{y}\|^2_2~\mbox{s.t.}~
\tilde{Z}=W.
\end{align}
The augmented Lagrangian function is provided by
\begin{align}\label{eq.43}
L_{\rho}(\tilde{Z},W,Y)=\lambda\|W\|^p_p+\frac{1}{2}\|\hat{A}\mbox{vec}(\tilde{Z})-\hat{y}\|^2_2+
<Y,\tilde{Z}-W>+\frac{\rho}{2}\|\tilde{Z}-W\|^2_F,
\end{align}
where $Y\in\mathbb{R}^{m\times n}$ is dual variable, and $\rho>0$ is a penalty parameter. Then, ADMM used to (\ref{eq.43}) comprises of the iterations as follows
\begin{align}
\label{eq.44}\tilde{Z}^{k+1}&=\arg\min_{\tilde{Z}}\frac{1}{2}\|\hat{A}\mbox{vec}(\tilde{Z})-\hat{y}\|^2_2+
\frac{\rho}{2}\|\tilde{Z}-(W^k-\frac{Y^k}{\rho})\|^2_F,\\
\label{eq.45}W^{k+1}&=\arg\min_{W}\lambda\|W\|^p_p+\frac{\rho}{2}\|\tilde{Z}^{k+1}-(W-\frac{Y^k}{\rho})\|^2_F,\\
\label{eq.46}Y^{k+1}&=Y^k+\rho(\tilde{Z}^{k+1}-W^{k+1}).
\end{align}
All solving processes are concluded in Algorithm 4.1.

\begin{algorithm}[H]
\label{alg.1}\caption{: Solve problem (\ref{eq.7}) by ADMM } 
\begin{algorithmic}[1]
\State Input $A\in\mathbb{R}^{M\times mn}$, $y\in\mathbb{R}^M$, perturbation $E\in\mathbb{R}^{M\times mn}$ with $\|E\|=\epsilon_A\|A\|$, $p\in(0,1]$.
\State Initialize $\hat{\mathcal{Z}}^0=W^0=Y^0$, $\gamma=1.1$, $\lambda_0 = 10^{-6}$, $\lambda_{\max} = 10^{10}$, $\rho= 10^{-6}$, $\varepsilon=10^{-8}$, $k=0$.
\While{not converged} 
　　\State Updated $\tilde{Z}^{k+1}$ by $$\tilde{z}=(\hat{A}^{\top}\hat{A}+\rho I)^{-1}\left(\hat{A}^{\top}\hat{y}+\rho\mbox{vec}(W^k)-\mbox{vec}(Y^k)\right);$$
     $\tilde{Z}^{k+1}\leftarrow\tilde{z}$: reshape $\tilde{z}$ to the matrix $\tilde{Z}^{k+1}$ of size $m\times n$.
    \State Update $W^{k+1}$ by
    $$\arg\min_{W}\lambda\|W\|^p_p+\frac{\rho}{2}\|\tilde{Z}^{k+1}-(W-\frac{Y^k}{\rho})\|^2_F;$$
     \State Update $Y^{k+1}$ by
     $$Y^{k+1}=Y^k+\rho(\tilde{Z}^{k+1}-W^{k+1});$$
     \State Update $\lambda_{j+1}$ by $\lambda_{j+1}=\min(\gamma\lambda_{j},\lambda_{\max});$
     \State Check the convergence conditions
      $$\|\tilde{Z}^{k+1}-\tilde{Z}^{k}\|_{\infty}\leq\varepsilon,~\|W^{k+1}-W^{k}\|_{\infty}\leq\varepsilon,$$
      $$\|\hat{A}\mbox{vec}(\tilde{Z}^{k+1})-\hat{y}\|_{\infty}\leq\varepsilon,~\|\tilde{Z}^{k}-W^{j+1}\|_{\infty}\leq\varepsilon.$$

\EndWhile
\end{algorithmic}
\end{algorithm}

In our experiments, we generate a measurement matrix $A\in\mathbb{R}^{M\times mn}$ with i.i.d. Gaussian $\mathcal{N}(0,1/M)$ elements. We generate $X\in\mathbb{R}^{m\times n}$ of rank $r$ by $X=PQ$, where $P\in\mathbb{R}^{m\times r}$ and $Q\in\mathbb{R}^{r\times n}$ are with its elements being zero-mean, one-variation Gaussian, i.i.d. random variables. We select $M=660$, $m=n=30$ and $r=0.2m$. With $X$ and $A$, the measurements $y$ are produced by $y=A\mbox{vec}(X)+z$, where $z$ is the Gaussian noise. The perturbation matrix $E$ is with its entries following Gaussian distribution, which fulfills $\|E\|=\epsilon_A\|A\|$, where $\epsilon_A$ denotes the perturbation level of $A$ and its value is not fixed. The perturbed matrix $\hat{A}$, $\hat{A}=A+E$, is used in (\ref{eq.44}). To avoid the randomness, we perform 100 times independent trails as well as the average result in all test.

To look for a proper parameter $\lambda$ that derives the better recovery effect, we carry out two sets of trails. Fig. \ref{fig.2} (a) and (b) respectively plot the parameter $\lambda$ and relative error (RelError, $\|X-X^*\|_F/\|X\|_F$) results in different $p$ values and perturbation level $\epsilon_A$. $\lambda$ ranges from $10^{-6}$ to $1$. Fig. \ref{fig.2} (a) and (b) show that $\lambda\in[10^{-6},10^{-2}]$ is relatively suitable.

\begin{figure}[h]
\begin{center}
\subfigure[]{\includegraphics[width=0.4544\textwidth,height=0.349\textwidth]{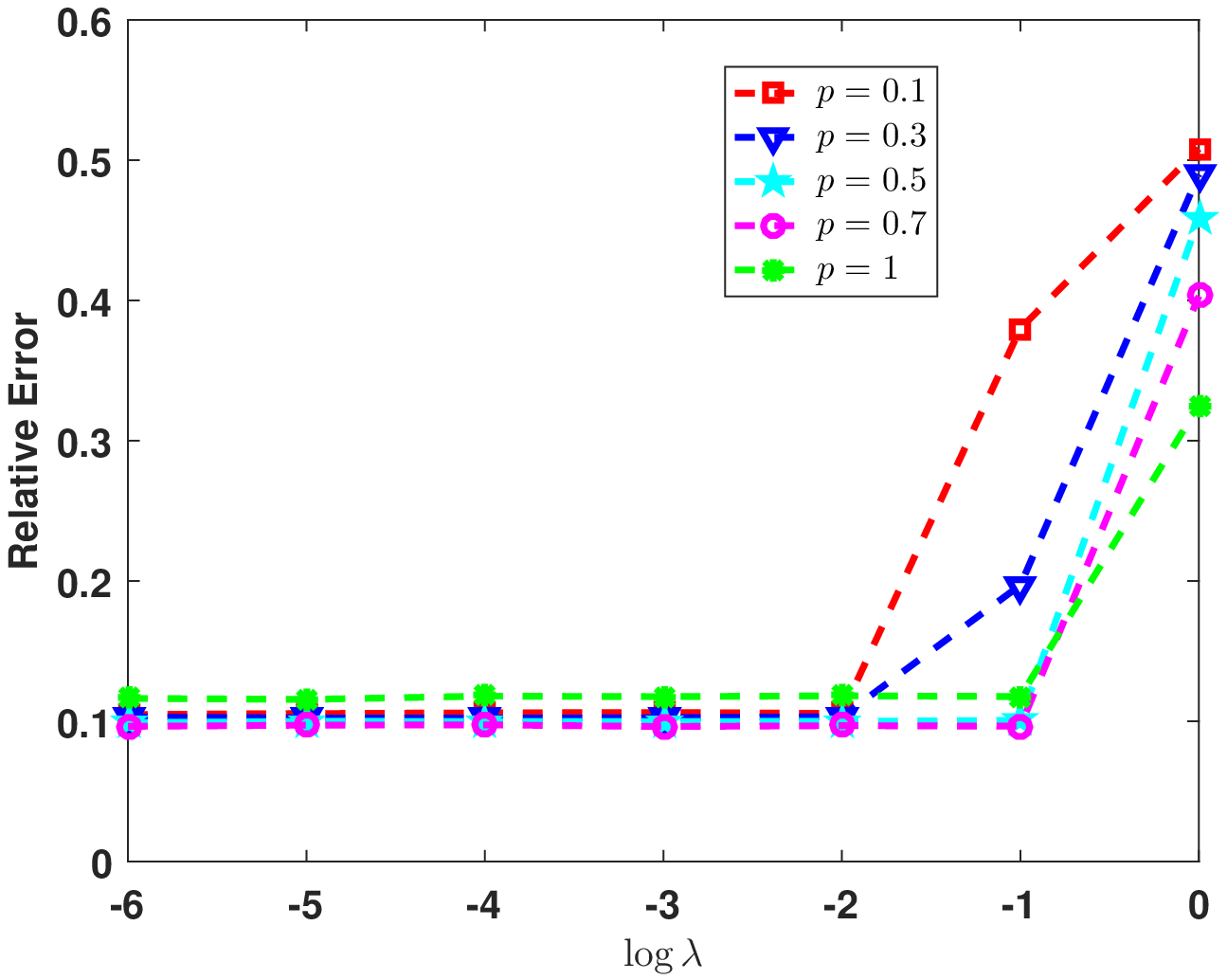}}
\subfigure[]{\includegraphics[width=0.4544\textwidth,height=0.349\textwidth]{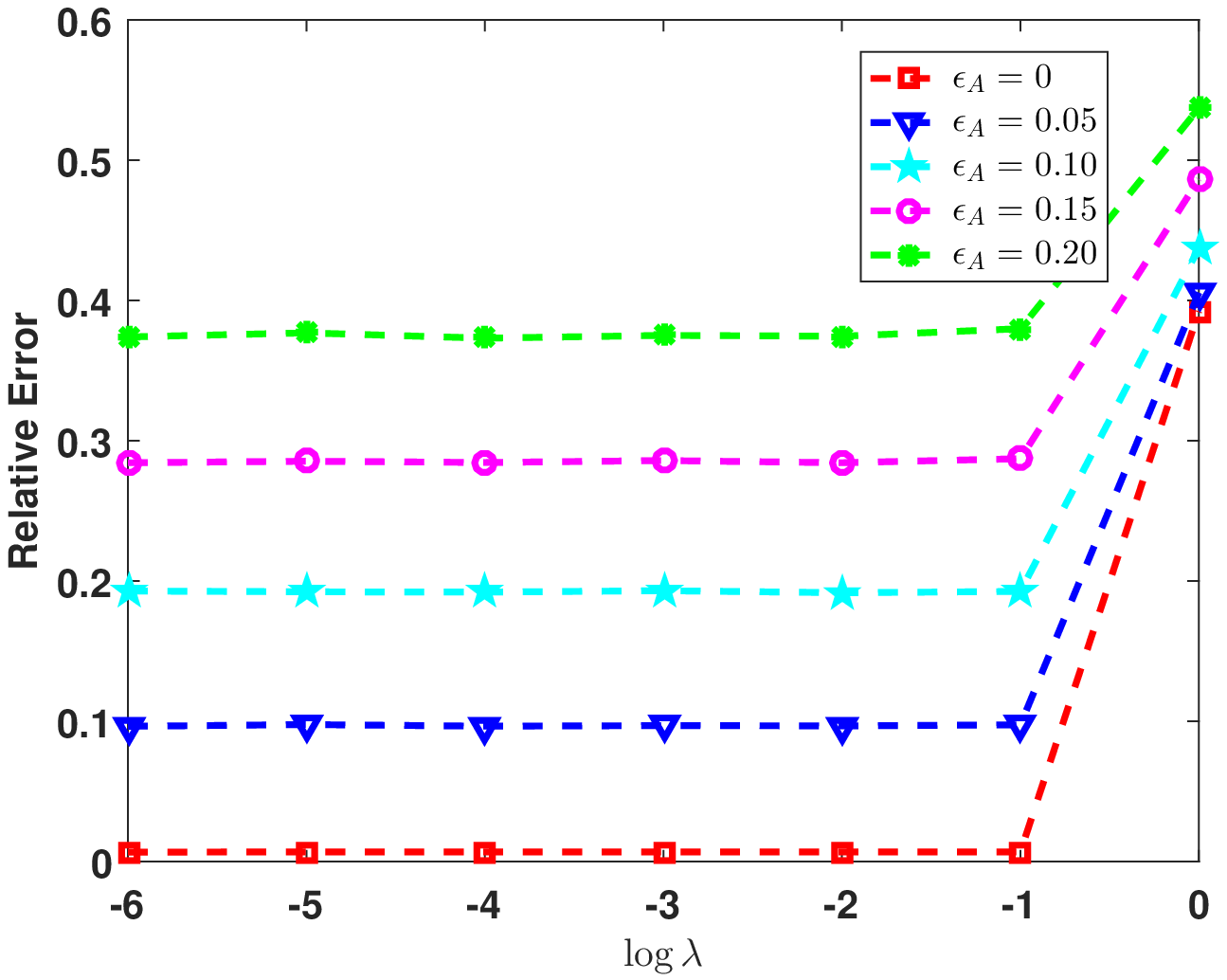}}
\caption{Parameter selection $\lambda$ for (a) $\epsilon_A=0.05$, (b) $p=0.7$}\label{fig.2}
\end{center}
\vspace*{-14pt}
\end{figure}

By giving $\lambda=10^{-6}$, we consider the convergence of Algorithm 4.1. Fig. \ref{fig.3}(a) presents the relationship between relative neighboring iteration error (RNIE, $r(k)=\|X^{k+1}-X^k\|_F/\|X^k\|_F$) and number of iterations $k$. One can easily see that with the increasing of iterations, RNIE decreases quickly, and when $k\geq250$, $r(k)<10^{-4}.$ The results that relative error versus the values of $p$ in different $\epsilon_A$ are showed in Fig. \ref{fig.3}(b). Fig. \ref{fig.3}(b) indicates that the proper choice of the size of $p$ will be helpful to facilitate the performance of nonconvex Schatten $p$-minimization.

\begin{figure}[h]
\begin{center}
\subfigure[]{\includegraphics[width=0.4344\textwidth,height=0.349\textwidth]{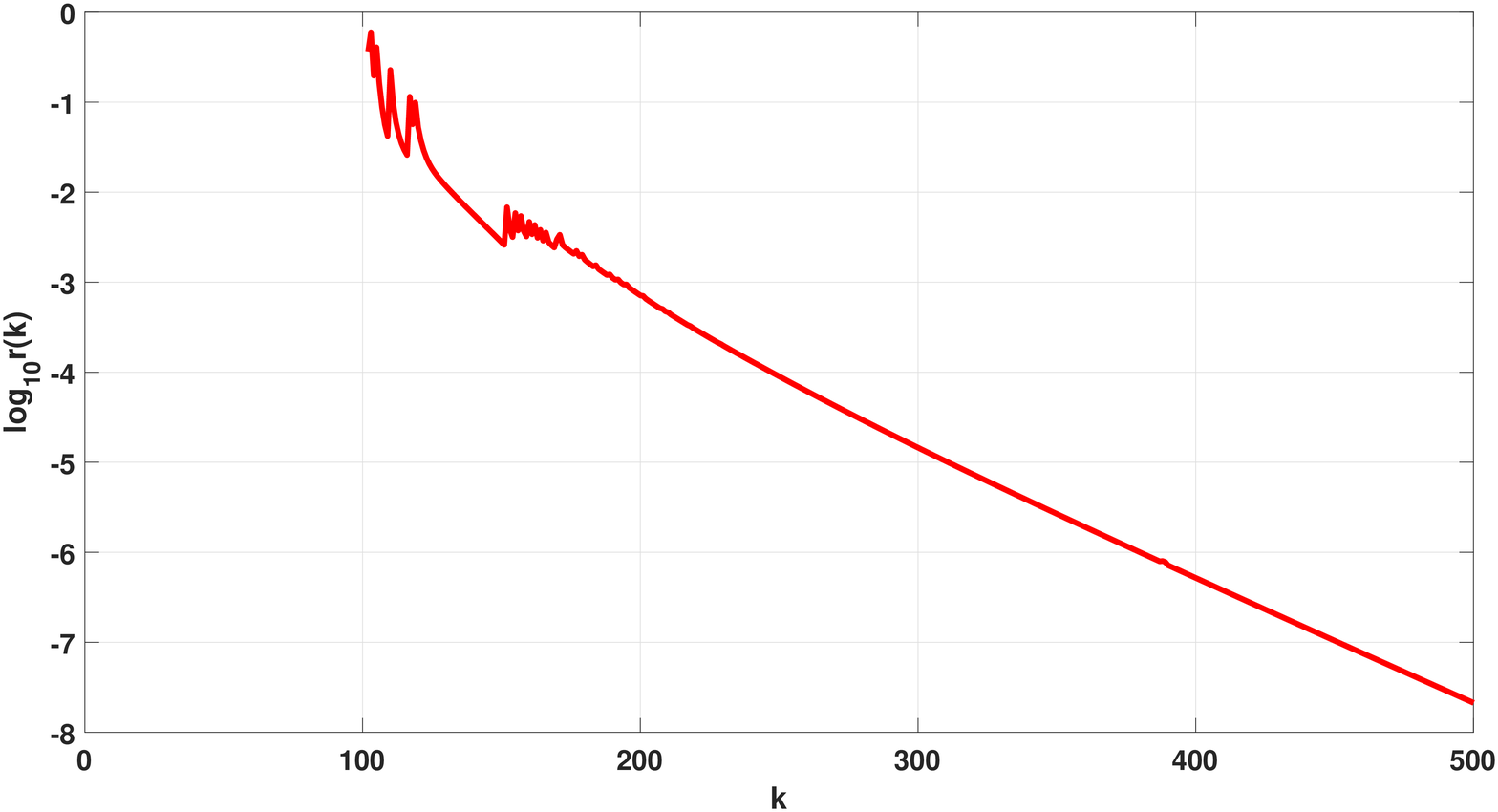}}
\hspace{0.5cm}
\subfigure[]{\includegraphics[width=0.4544\textwidth,height=0.349\textwidth]{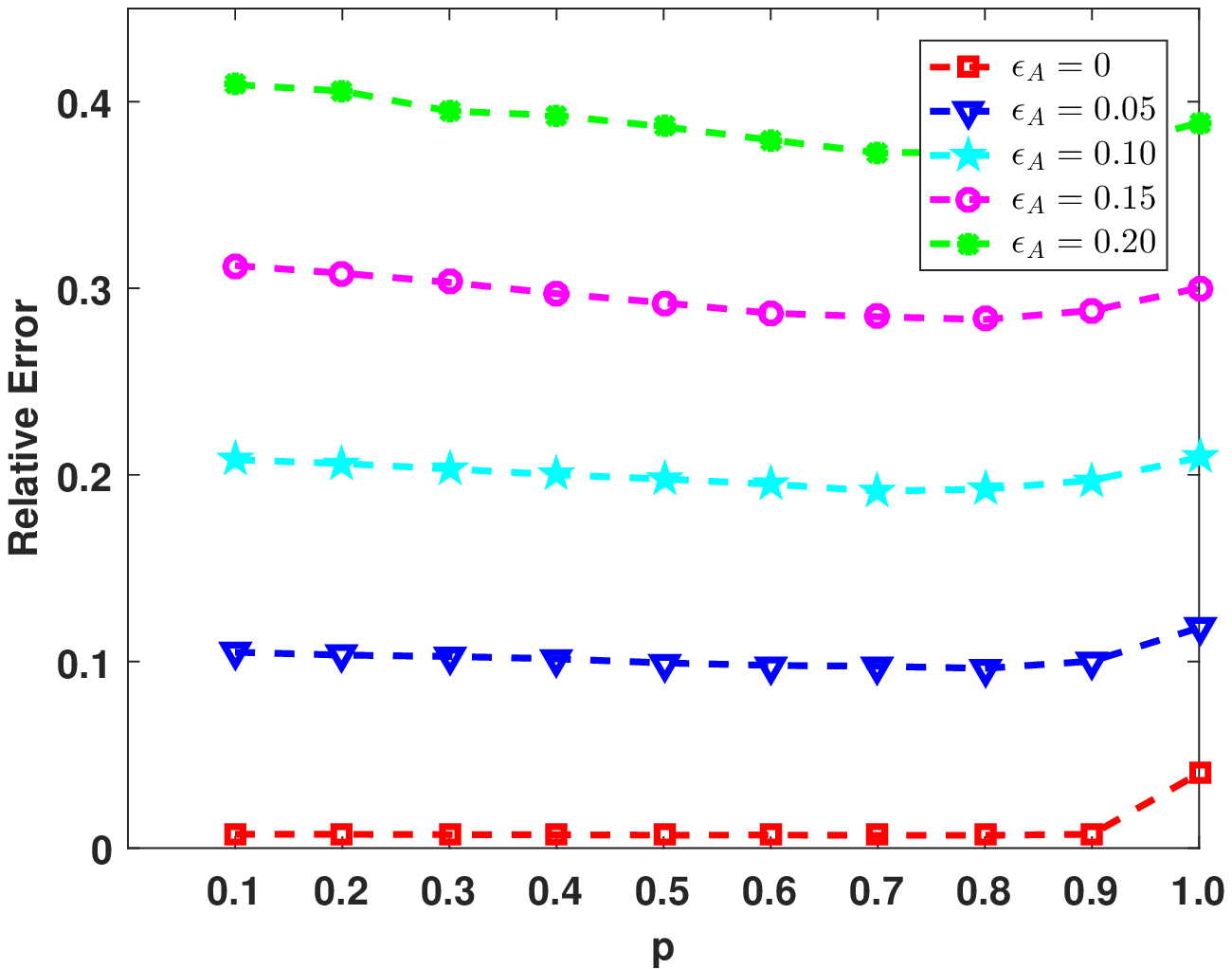}}
\caption{Convergence analysis for Algorithm 4.1 and relative error versus $p$. (a) Convergence analysis, $\epsilon_A=0.05$, (b)Reconstruction performance of completely perturbed nonconvex Schatten $p$-minimization, varying $p$}\label{fig.3}
\end{center}
\vspace*{-14pt}
\end{figure}

The theoretical error bound and $\|X-X^*\|^p_F$ versus the values of $p$ with $a=2$, $\delta_{2ar}=\delta_{(a+1)r}=0.1$ and $r=6$ in different perturbation level $\epsilon_A$, the results are provided in Figs. \ref{fig.7} (a) and (b). The values of $p$ vary from $0.1$ to $0.9$. From the observation of Fig. \ref{fig.7}, $\|X-X^*\|^p_F$ is smaller than the theoretical error bound.

\begin{figure}[h]
\begin{center}
\subfigure[]{\includegraphics[width=0.4344\textwidth,height=0.349\textwidth]{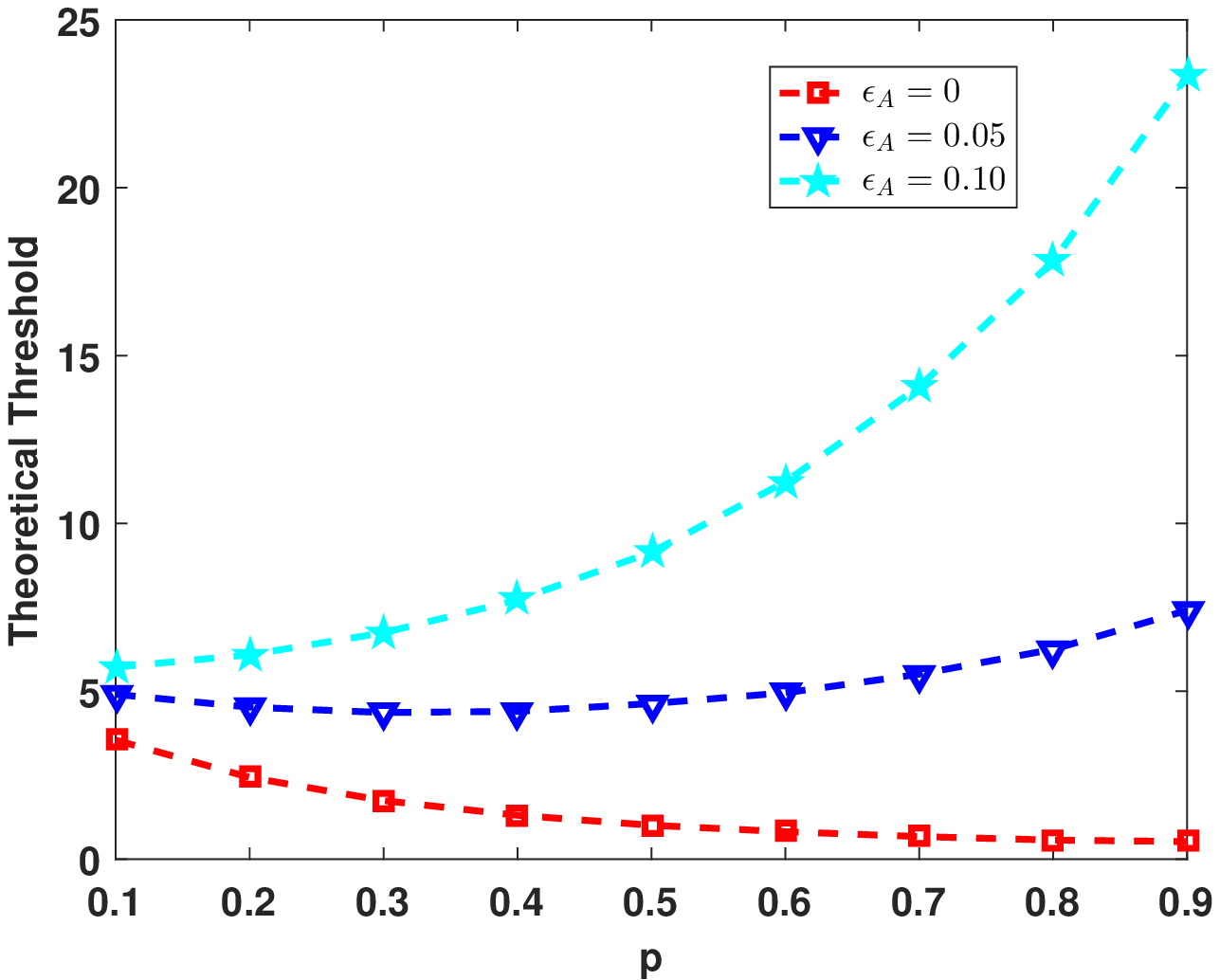}}
\hspace{0.5cm}
\subfigure[]{\includegraphics[width=0.4544\textwidth,height=0.349\textwidth]{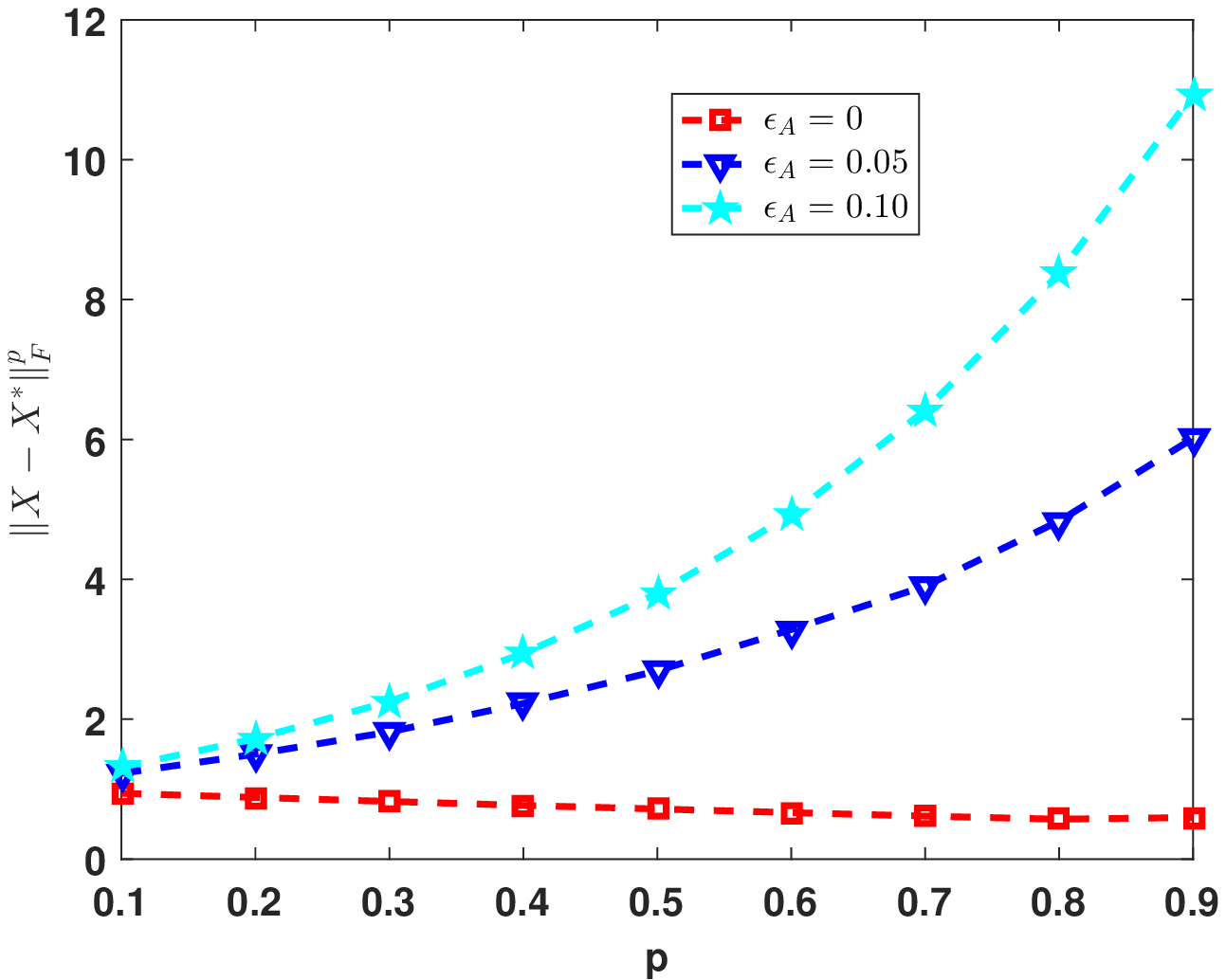}}
\caption{Theoretical error bound and $\|X-X^*\|^p_F$ versus $p$ for (a) $a=2$, $\delta_{2ar}=\delta_{(a+1)r}=0.1$, (b) $r=6$.}\label{fig.7}
\end{center}
\vspace*{-14pt}
\end{figure}

In Fig. \ref{fig.4}, the relative error is plotted versus the number of measurements $M$ in different $\epsilon_A=0,~0.05,~0.10,~0.15,~0.20$ and $p=0.1,~0.3,~0.5,~0.7,~1$, respectively. From Fig. \ref{fig.4}, with the increase of number of measurements or the decrease of perturbation level, the recovery performance of nonconvex Schatten $p$-minimization gradually improves. Moreover, Fig. \ref{fig.4}(b) reveals that the performance of nonconvex Schatten $p$-minimization is better than that of convex nuclear norm minimization. In Fig. \ref{fig.5}, we plot the the relative error versus the rank $r$ of the matrix $X$ for different $\epsilon_A$ and $p$, respectively. The results indicate that the smaller the rank of the matrix, the better the recovery performance, and choosing a smaller perturbation level or the values of $p$ will improve the reconstruction effect of nonconvex Schatten $p$-minimization.

\begin{figure}[h]
\begin{center}
\subfigure[]{\includegraphics[width=0.4344\textwidth,height=0.349\textwidth]{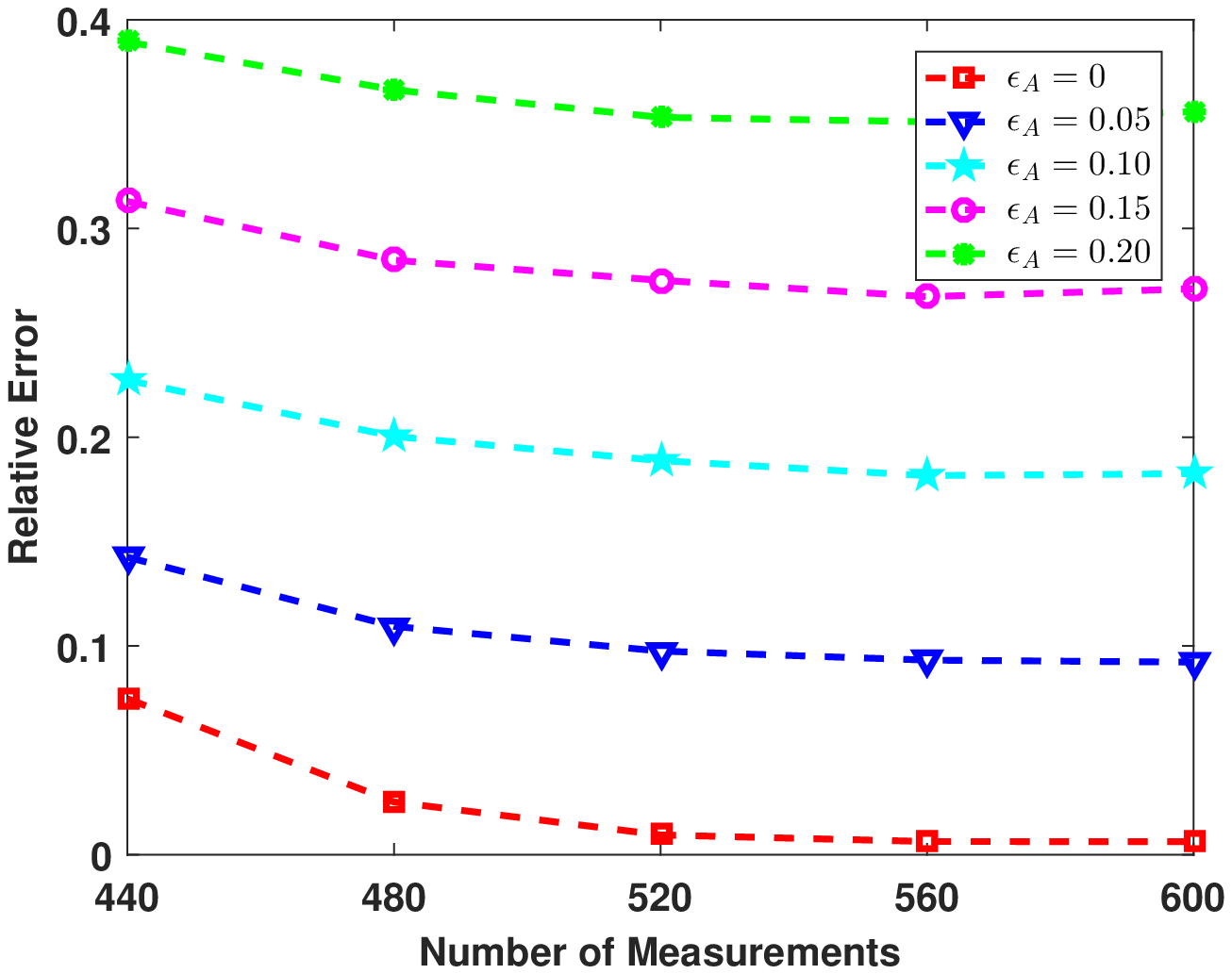}}
\hspace{0.5cm}
\subfigure[]{\includegraphics[width=0.4544\textwidth,height=0.349\textwidth]{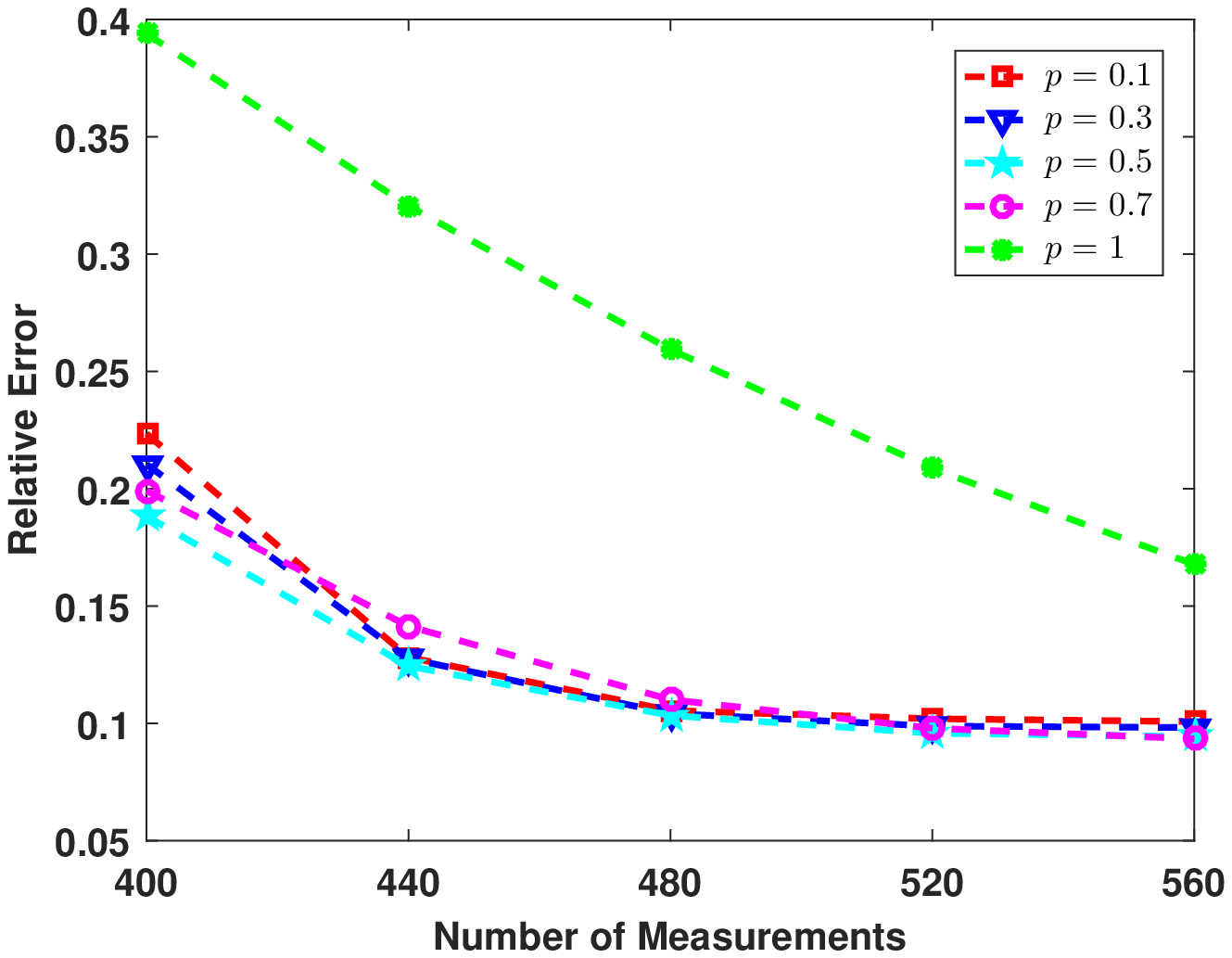}}
\caption{Reconstruction performance of completely perturbed nonconvex Schatten $p$-minimization versus number of measurements $M$. (a) $p=0.7$,  (b)$\epsilon_A=0.05$}\label{fig.4}
\end{center}
\vspace*{-14pt}
\end{figure}

\begin{figure}[h]
\begin{center}
\subfigure[]{\includegraphics[width=0.4344\textwidth,height=0.349\textwidth]{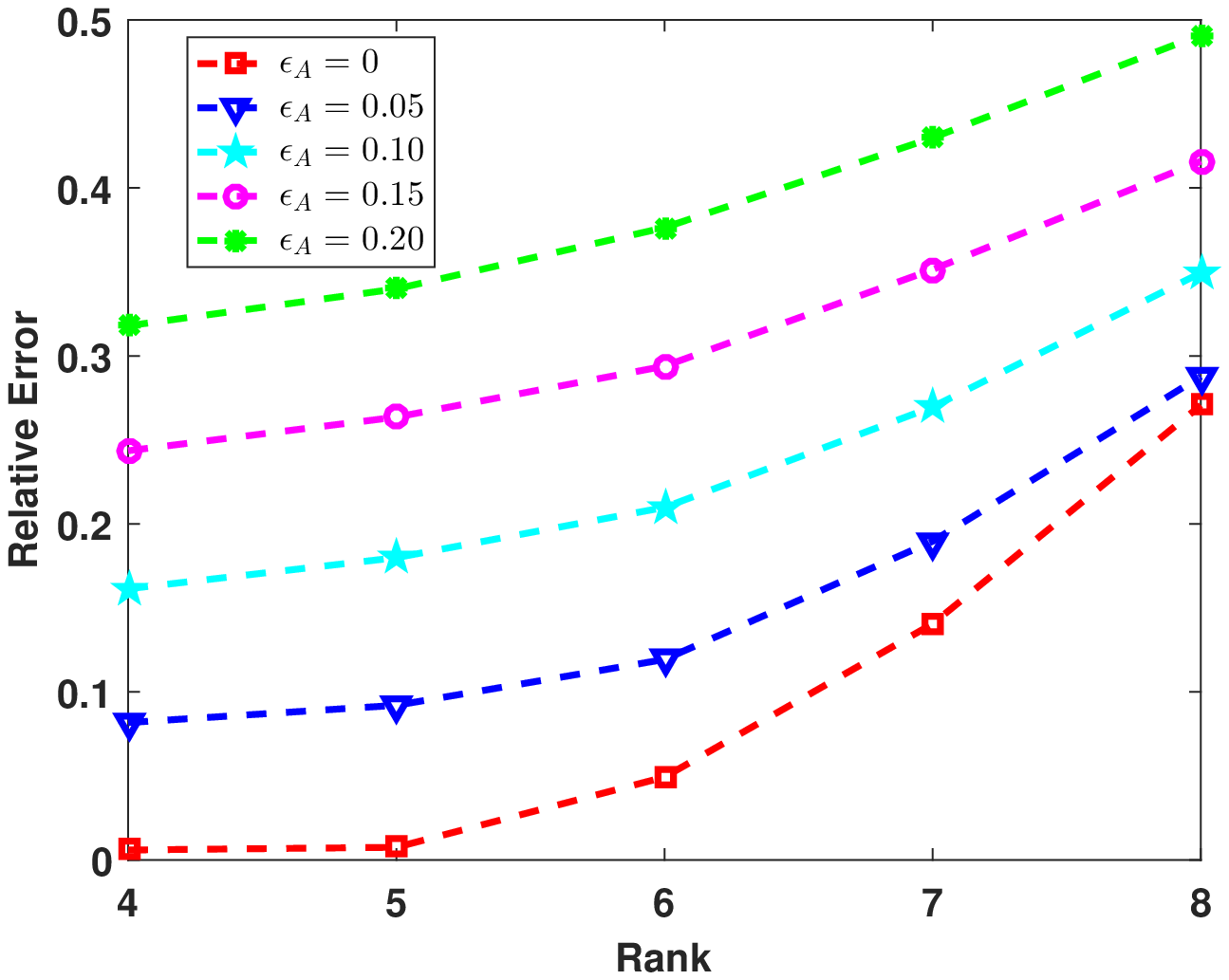}}
\hspace{0.5cm}
\subfigure[]{\includegraphics[width=0.4544\textwidth,height=0.349\textwidth]{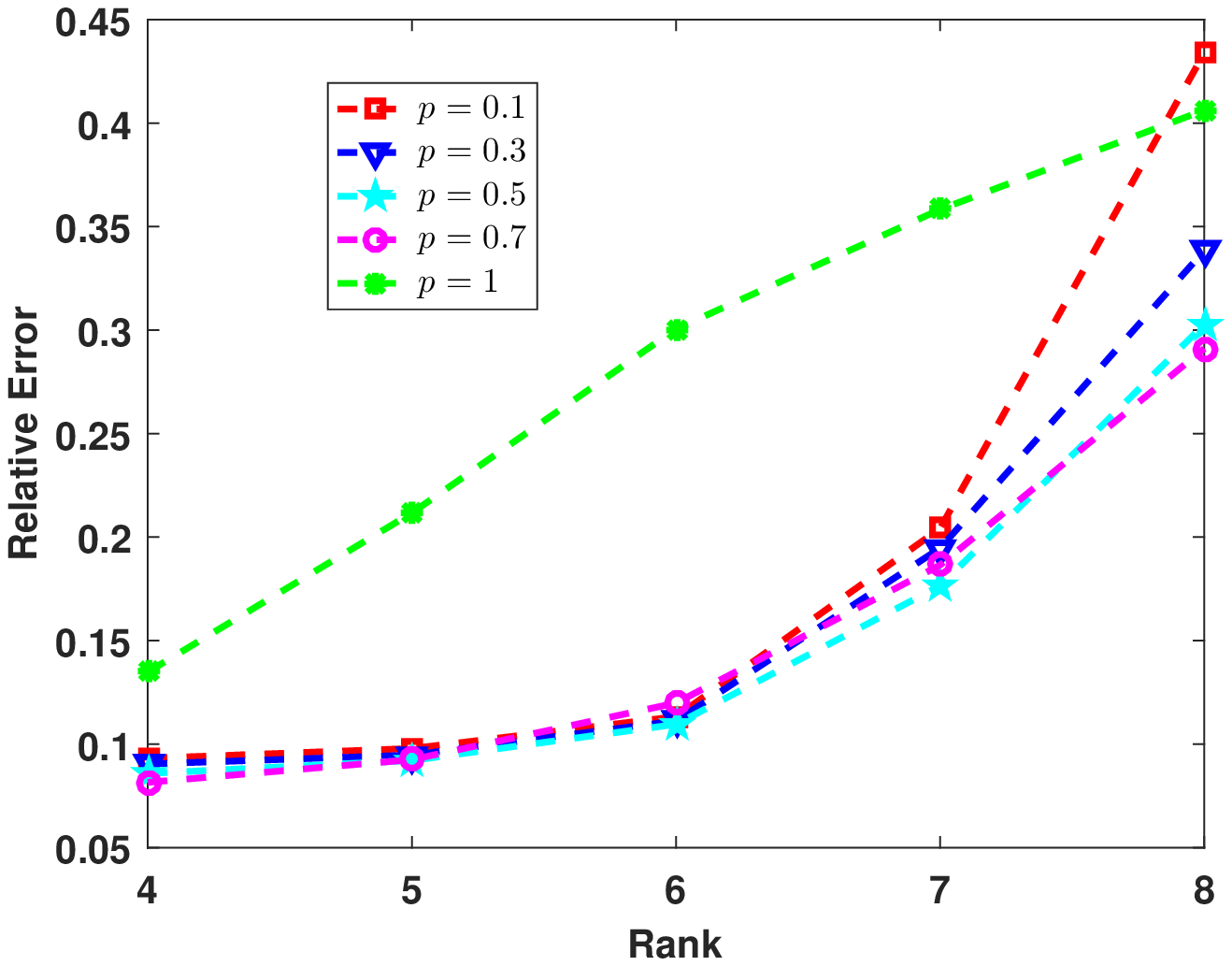}}
\caption{Reconstruction performance of completely perturbed nonconvex Schatten $p$-minimization varying rank $r$. (a) $p=0.7$,  (b) $\epsilon_A=0.05$}\label{fig.5}
\end{center}
\vspace*{-14pt}
\end{figure}

Furthermore, Fig. \ref{fig.6} offers the results concerning the recovery performance of the nonconvex method and the convex method for the $\epsilon_A=0.05$. The curves of relationship between the relative error and the rank $r$ are described by nonconvex Schatten $p$-minimization and convex nuclear norm minimization, respectively. Fig. \ref{fig.6} displays that the performance of nonconvex method is superior to that of the convex method.

\begin{figure}[h]
\begin{center}
{\includegraphics[width=0.50\textwidth]{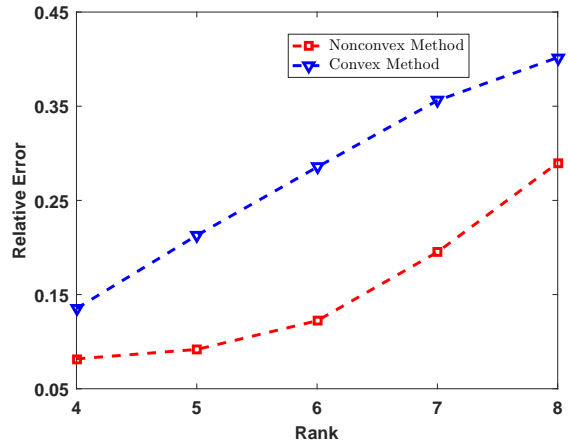}}
\end{center}
\vspace*{-14pt}
\caption{Reconstruction performance of nonconvex Schatten $p$-minimization and convex nuclear norm minimization, varying rank $r$ for $\epsilon_A=0.05$}\label{fig.6}
\end{figure}

\section{Conclusion}

In this paper, we investigate the completely perturbed problem employing the nonconvex Schatten $p$-minimization for reconstructing low-rank matrices. We derive a sufficient condition and the corresponding upper bounds of error estimation. The gained results reveal the nonconvex Schatten $p$-minimization has the stability and robustness for reconstructing low-rank matrices with the existence of a total noise. The practical meaning of gained results, not only can conduct the choice of the linear transformations for reconstructing low-rank matrices, that is, a linear transformation with a smaller RIC instead of a larger one can superior enhance the reconstruction performance, but also can also present a theoretical sustaining to approximation accurateness. Moreover, the numerical experiments further show the verification of our results, and the performance of nonconvex Schatten $p$-minimization is better than that of convex nuclear norm minimization in the complete perturbation situation.


\end{document}